\documentclass[12pt, draftclsnofoot, onecolumn]{IEEEtran}

\usepackage{graphicx}
\usepackage{amssymb,amsmath}
\usepackage{multirow}
\usepackage{tabularx}
\usepackage{slashbox}
\usepackage{color}
\usepackage{lettrine}
\usepackage{mdwlist}
\usepackage{cite}
\newtheorem{Proposition}{Proposition}

\begin{document}
	\title{Performance Analysis of SPAD-Based Optical Wireless Communication with OFDM}
	\author{Shenjie Huang, Yichen Li, Cheng Chen, Mohammad Dehghani Soltani, Robert Henderson, Majid Safari, and Harald Haas 
		\thanks{S. Huang and Y. Li contribute equally to this work. S. Huang, M. D. Soltani, R. Henderson and M. Safari are with the School of Engineering, the University of Edinburgh, Edinburgh EH9 3JL, UK. C. Chen and H. Haas are with LiFi Research and Development Centre, University of Strathclyde, Glasgow G1 1RD, UK. (\textit{Corresponding author: S. Huang.})}
	}
	\maketitle
	\begin{abstract}
	 In recent years, there has been a growing interest in the use of single-photon avalanche diode (SPAD) in optical wireless communication (OWC). SPAD operates in the Geiger mode and can act as a photon counting receiver obviating the need for a transimpedance amplifier (TIA). Although a SPAD receiver can provide higher sensitivity compared to the traditional linear photodetectors, it suffers from the dead-time-induced nonlinearity. To improve the data rates of SPAD-based OWC systems, optical orthogonal frequency division multiplexing (OFDM) can be employed.  This paper provides a comprehensive theoretical analysis of the SPAD-based OWC systems using OFDM signalling considering the effects of signal clipping, SPAD nonlinearity, and signal-dependent shot noise. An equivalent additive Gaussian noise channel model is proposed to describe the performance of the SPAD-based OFDM system. The statistics of the proposed channel model and the analytical expressions of the signal-to-noise ratio (SNR) and bit error rate (BER) are derived in closed forms.  By means of extensive numerical results, the impact of the unique receiver nonlinearity on the system performance is investigated. The results demonstrate new insights into different optical power regimes of reliable operation for SPAD-based OFDM systems even well beyond SPAD saturation level.      
	\end{abstract}

\section{Introduction}
Due to the scarcity of the radio frequency (RF) spectrum, optical wireless communication (OWC) has attracted significant interest in both industry and scientific communities in recent years. Compared to traditional RF wireless communication, the potential advantages of OWC mainly include high data rates, excellent security levels and license-free spectrum \cite{Khalighi14}. However, the performance of OWC can be strongly degraded by the occasional outages caused by the received optical power fluctuations. On the one hand, for free-space optical (FSO) communication, the power fluctuation is mainly introduced by atmospheric turbulence, misalignment and adverse weather conditions, e.g., fog and haze \cite{HuangHy}. On the other hand, for visible light communication (VLC) the power fluctuation mainly results from user mobility and random orientation \cite{Soltani19}. One effective way of mitigating the effect of outages induced by power fluctuation is employing a highly sensitive photon counting receivers such as a single-photon avalanche diode (SPAD). The photon counting capability of a SPAD is achieved by biasing the traditional linear photodiode above the breakdown voltage so that it operates in the Geiger mode. Compared to the commonly used photodetectors, e.g., PIN photodiode (PD) and avalanche photodiode (APD), a SPAD obviates the need for a transimpedance amplifier (TIA) and has the advantage of higher sensitivity which can greatly extend the operation regime of the received optical power close to the quantum limit \cite{Zimmermann19}. In the literature, many studies have considered the applications of SPAD receivers in VLC \cite{Chitnis14,Zhang:18}, FSO \cite{HuangHy}, and underwater wireless optical communication (UWOC) \cite{Khalighi20}. 

Although SPAD-based receivers can provide single photon sensitivity, after each avalanche triggered by a photon detection, the SPAD needs to be quenched when it becomes blind to any incident photon arrivals for a short period of time called \textit{dead time}. It is known that the performance of SPAD receivers is greatly influenced by the dead-time-induced nonlinearity \cite{Eisele11,HuangHy}. To improve the performance of SPAD-based receivers, a detection scheme which uses not only the photon count information but also the photon arrival information is proposed in \cite{Huang20}. In addition, a SPAD-based OWC system with joint pre-distortion
and noise normalization functionality is proposed to mitigate the SPAD nonlinear distortion \cite{Huang22}. Although most of the SPAD-based OWC works focus on simple on-off keying (OOK), optical orthogonal frequency division multiplexing (OFDM) has also been employed in SPAD-based systems to improve the spectral efficiency (SE) \cite{Zhang22,Huang:22,Li15}. In particular, in \cite{Huang:22}, an experimental data rate of $5$ Gbps using a commercial SPAD receiver is reported, which is significantly higher than the data rates achieved with OOK \cite{Ahmed20,Matthews21}. Despite a few experimental works, to the best of our knowledge, a theoretical performance analysis of SPAD-based OFDM OWC systems in the presence of nonlinear effects at both transmitter and receiver is still missing. Therefore, the objective of this paper is to fill this research gap.    

The performance analysis of OFDM-based OWC systems with traditional linear photodetectors has been widely investigated. In particular, in \cite{Dimitrov12}, the impact of clipping noise on the performance of OWC systems with OFDM is discussed. Note that the time-domain OFDM signal is approximately normal distributed with high peak-to-average power ratio (PAPR), but practical light sources have limited dynamic ranges. Therefore, signal clipping should be employed to fit the signal into the limited dynamic ranges.  In \cite{Tsonev13}, a set of polynomials is employed to further explore the effects of a general nonlinear distortion on OWC systems. Different from the systems with linear receivers, when SPAD is employed in OWC systems, the unique receiver nonlinearity induced by the SPAD dead time has to be also taken into account. In addition, being a photon counting receiver, the dominant noise factor of a SPAD receiver is the signal-dependent shot noise rather than the signal-independent thermal noise. As a result, the performance analysis of traditional OFDM with linear PD cannot be directly applied to OFDM with a SPAD receiver.   

This work provides a complete framework for the performance analysis of OFDM-based OWC systems with SPAD receivers considering practical nonlinear effects at both transmitter (i.e., clipping) and receiver (i.e., dead time induced nonlinearity) sides. An equivalent additive Gaussian noise channel model is defined to characterize the performance of the SPAD-based OFDM OWC systems. The statistics of the proposed channel is studied analytically and is given in closed form. The analytical expressions of the system signal-to-noise ratio (SNR) and bit error rate (BER) are derived. It is demonstrated that the SPAD nonlinearity has a significant impact on the communication performance of the OFDM system. The results provide new insights into the different regimes of reliable operation of SPAD receivers showing that high-speed transmission is achievable even at high power regimes beyond the saturation of SPAD receivers. Moreover, the trade-off between the received electrical signal power and nonlinear distortion indicates that the signal clipping at the transmitter can be optimized to improve the performance.         

The rest of this paper is organized as follows. Section \ref{OFDMsys}
introduces the SPAD-based OWC system with OFDM. Section \ref{theore}  presents the theoretical analysis of such system. The results and discussion are
presented in Section \ref{NumRes}. Finally, we conclude this paper
in Section \ref{con}.
 
\section{SPAD-Based OFDM System}\label{OFDMsys}
\begin{figure}[!t]
	\centering		\includegraphics[width=0.99\textwidth]{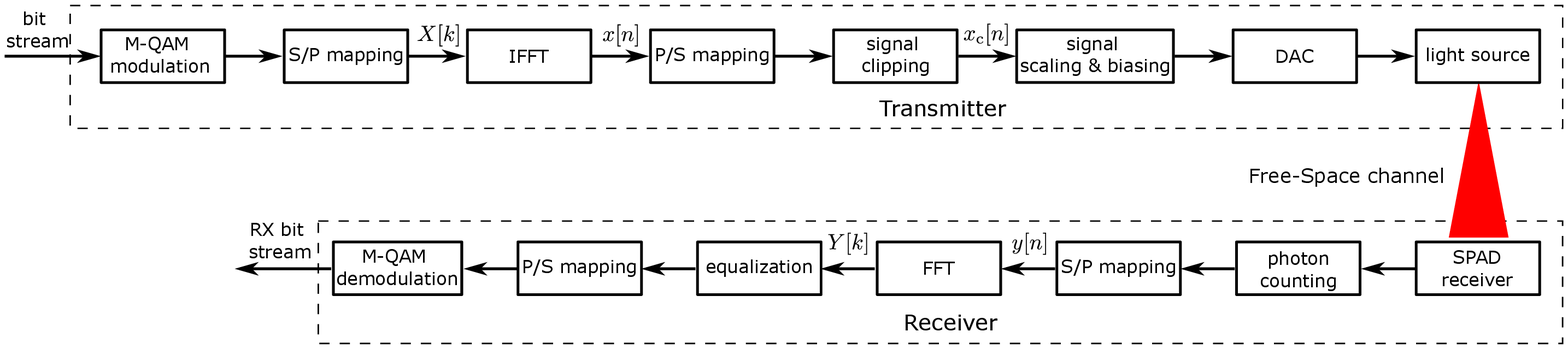}
	\caption{Block diagram of an OFDM OWC system with SPAD receiver.}\label{system_figure}
\end{figure}
\subsection{Optical OFDM Transmission}
Fig. \ref{system_figure} presents the block diagram of an OFDM OWC communication system using a SPAD-based optical receiver. In this work, we focus on the direct-current-biased optical OFDM (DCO-OFDM) due to its high SE and implementation simplicity, noting that similar analyses can be conducted for other OFDM schemes, e.g., asymmetrically clipped optical OFDM (ACO-OFDM). At the transmitter, the input bit stream is transformed into a complex symbol stream by the $M$-quadrature amplitude modulation (QAM) modulator, where $M$ denotes the constellation size. The symbol stream is then serial-to-parallel (S/P) converted to form vectors suitable for inverse fast Fourier transform (IFFT) operation. Consider a $K$-point fast Fourier transform (FFT) operation, the subcarriers carrying information bits populate the first half of the OFDM frame whereas leaving the $0$-th and $K/2$-th subcarriers unused \cite{Dimitrov13}. Therefore, the number of information carrying subcarriers is $K'=K/2-1$. Hermitian symmetry is applied to the rest of the OFDM frame in order to obtain real-valued symbols  after the  IFFT operation. Denote the symbol allocated to the $k$th information carrying subcarrier as $X[k]$ with $k = 1, \cdots,K'$.  Then, the time-domain signal $x[n]$ can be obtained after the IFFT as
\begin{equation}
	x[n]=\frac{1}{\sqrt{K}} \sum_{k=0}^{K-1} X[k] e^{\frac{2\pi nkj}{K}}. 
\end{equation}
According to the central limit theorem (CLT), the amplitude of $x[n]$ is approximately zero-mean Gaussian distributed as long as $K$ is relatively large \cite{Tsonev13}. Considering the uniform power alloction over the subcarriers, to ensure that $x[n]$ is with unit variance, the variance of $X[k]$ is set to $\sigma_X^2=K/(K-2)$ \cite{Dimitrov12}.

Considering that the PAPR of the generated signal $x[n]$ is relatively high whereas practical light sources are with limited dynamic ranges, $x[n]$ has to be properly clipped. The clipped signal can be expressed as 
\begin{equation}\label{xclipped}
	x_{\rm c}[n]=
	\begin{cases}
		\kappa_{\rm t}, & \mathrm{if}\qquad  x[n]\geq \kappa_{\rm t},\\
		x[n],     & \mathrm{if}\qquad  \kappa_{\rm b}<x[n]< \kappa_{\rm t},\\
		\kappa_{\rm b}, & \mathrm{if}\qquad  x[n]\leq \kappa_{\rm b}, 	
	\end{cases}
\end{equation}
where $\kappa_{\rm t}$ and $\kappa_{\rm b}$ denote the normalized clipping levels from the top and bottom, respectively. After scaling, adding a DC bias and digital-to-analog conversion, the resultant electrical signal is used to drive the light source. In effect, the optical power of the $n$th time-domain OFDM sample emitted from the source is given by \cite{Dimitrov12}
\begin{equation}\label{xt}
	x_{\rm t}[n]=\delta\, x_{\rm c}[n]+P_{\rm bias},
\end{equation}
where $\delta$ denotes the scaling factor and $P_{\rm bias}$ is the DC bias value. Since the light source can only send unipolar signals, $x_{\rm t}[n]$ should be non-negative by choosing proper DC bias. Denoting the minimal and maximal optical power of the light source as $P_{\rm min}$ and $P_{\rm max}$, respectively, the following equations should be satisfied: $\delta \kappa_{\rm b}+P_{\rm bias}=P_{\rm min}$ and $\delta \kappa_{\rm t}+P_{\rm bias}=P_{\rm max}$, which leads to 
\begin{equation}
	\delta=\frac{P_{\rm max}-P_{\rm min}}{\kappa_{\rm t}-\kappa_{\rm b}}, \,\, {\rm and} \,\, P_{\rm bias}=\frac{P_{\rm min}\kappa_{\rm t}-P_{\rm max}\kappa_{\rm b}}{\kappa_{\rm t}-\kappa_{\rm b}}.
\end{equation}
The average transmit optical power can be expressed as \cite{Dimitrov12}
\begin{equation}
	\overline{P}_{\rm Tx}=\delta\left[f(\kappa_b)-f(\kappa_t)+\kappa_t Q(\kappa_t)+\kappa_bQ(-\kappa_b)\right]+P_{\rm bias},
\end{equation}
where 
\begin{equation}\label{fx}
f(x)=\frac{1}{\sqrt{2\pi}}\exp\left(-\frac{x^2}{2}\right),
\end{equation}
is the  probability density function (PDF) of a standard Gaussian distribution and ${Q}(\cdot)$ denotes the Q-function. Note that if a symmetric clipping is employed, i.e., $\kappa_{\rm t}=-\kappa_{\rm b}=\kappa$, one has $\delta=\left(P_{\rm max}-P_{\rm min}\right)/2\kappa$, $P_{\rm bias}=\left(P_{\rm max}+P_{\rm min}\right)/2$, and $\overline{P}_{\rm Tx}=P_{\rm bias}$. 

\subsection{SPAD Receivers}
A SPAD is an APD which is biased beyond reverse breakdown voltage in the so-called `Geiger' region. In this mode of operation, a SPAD triggers billions of electron-hole pair generations for each detected photon. In other words, a SPAD emits a very large current by receiving a single photon and thus can essentially be modelled as a single photon counter. The photodetection process of an ideal photon counter can be modelled using Poisson statistics which describe the shot noise effects \cite{Safari12}. However, the performance of a practical SPAD-based receivers depends on non-ideal effects such as dead time, photon detection efficiency (PDE), dark count rate (DCR), afterpulsing and crosstalk. In particular, due to the existence of dead time, significant nonlinearity is introduced which causes the deviation of the detected photon count of SPAD receivers from the Poisson distribution \cite{Huang22}. In order to mitigate the dead time effect and hence improve the photon counting capability, an array of SPADs whose output is the superposition of the photon counts of the individual SPADs is commonly utilized \cite{Fisher13,Huang20}. 

Assuming that the channel loss is $\zeta$, the average received signal optical power is given by $\overline{P}_{\rm Rx}=\zeta\,\overline{P}_{\rm Tx}$. The received optical power when the $n$th OFDM sample is transmitted can be expressed as $P_{\rm Rx}[n]=\zeta x_{\rm t}[n]$. To the SPAD receiver, it corresponds to an incident photon rate of \cite{Khalighi20}
\begin{equation}\label{lama}
	\lambda_\mathrm{a}[n]=\left(\frac{\Upsilon_{\rm PDE}\zeta x_{\rm t}[n]}{E_{\rm ph}}+\vartheta_{\rm DCR}+\vartheta_{\rm B}\right)\left(1+\varphi_{\rm AP}+\varphi_{\rm CT}\right),
\end{equation}
where $\Upsilon_{\rm PDE}$ denotes the PDE of the SPAD, $E_{\rm ph}$ is the photon energy, $\vartheta_{\rm B}$ represents the background photon rate, and $\vartheta_{\rm DCR}$, $\varphi_{\rm AP}$ and $\varphi_{\rm CT}$ denote the DCR of the array, the probabilities of afterpulsing and crosstalk, respectively. The photon rate $\vartheta_{\rm B}$ equals to $\Upsilon_{\rm PDE}P_{\rm B}/E_{\rm ph}$ where $P_{\rm B}$ denotes the power of the background light received by the array receiver. Note that in this work, it is assumed that the signalling bandwidth is less than the bandwidth of the light source, hence the channel frequency response is fairly flat over the signalling spectrum so that the channel-induced intersymbol interference is negligible \cite{Dimitrov12,Tsonev13}.  By denoting 
\begin{equation}
	\begin{cases}
		\mathcal{C}_{\rm s}=\Upsilon_{\rm PDE}\zeta\left(1+\varphi_{\rm AP}+\varphi_{\rm CT}\right)/E_{\rm ph},\\[2pt]
		\mathcal{C}_{\rm n}=(\vartheta_{\rm DCR}+\vartheta_{\rm B})\left(1+\varphi_{\rm AP}+\varphi_{\rm CT}\right),
	\end{cases}
\end{equation}
$\lambda_\mathrm{a}[n]$ can be simplified as
\begin{equation}\label{lama2}
	\lambda_\mathrm{a}[n]= \mathcal{C}_{\rm s}x_{\rm t}[n]+\mathcal{C}_{\rm n}.
\end{equation}

{The channel loss $\zeta$ depends on the specific OWC applications. For example, for FSO application the channel loss is random which is given by \cite{Jamali16}
	\begin{equation}\label{alphaFSO}
	\zeta_\mathrm{FSO}=h_fh_g=h_f \left[\mathrm{erf}\left(\frac{\sqrt{\pi}\omega}{2\sqrt{2}\theta L}\right)\right]^2,
	\end{equation}		
	where $\mathrm{erf}(\cdot)$ denotes the error function, $h_g$ refers to the geometric loss induced by diffraction of Gaussian beam, $h_f$ is the intensity fluctuation caused by atmospheric turbulence, $\omega$ is the receiver aperture diameter, $\theta$ represents the beam divergence angle at the transmitter, and $L$ is the link distance. The intensity fading $h_{f}$ is Gamma-Gamma distributed with PDF given by \cite{HuangHy}
	\begin{equation}
		\mathcal{F}_{h_{f}}(x)=\frac{2\left(\varrho \beta \right)^{\left(\varrho+\beta\right)/2}}{\Gamma(\varrho)\Gamma(\beta)}{x}^{\left(\varrho+\beta\right)/2-1} \mathcal{K}_{\varrho-\beta}\left(2\sqrt{\varrho\beta x}\right),
	\end{equation}
	where $\Gamma(\cdot)$ is the Gamma function, $\mathcal{K}_p(\cdot)$ is the modified Bessel function of the second kind, and the parameter $\varrho$ and $\beta$ are given by
	\begin{align}
		\varrho&=\bigg[\mathrm{exp}\bigg( \frac{0.49\chi^2}{\left(1+\!0.18\xi^2+0.56\chi^{12/5}\right)^{7/6}}\bigg)-1\bigg]^{-1},\nonumber\\
		\beta&=\bigg[\mathrm{exp}\bigg( \frac{0.51\chi^2\left(1\!+\!0.69\chi^{12/5}\right)^{-5/6}}{\left(1+\!0.9\xi^2+0.62\xi^2\chi^{12/5}\right)^{5/6}}\bigg)-1\bigg]^{-1},
	\end{align}
	respectively, with $\chi^2=0.5C_n^2k^{7/6}L^{11/6}$, $\xi^2=k_{\rm w}\omega^2/4L$ and $k_{\rm w}=2\pi/\lambda_\mathrm{op}$ where $\lambda_\mathrm{op}$ denotes the light wavelength and $C_n^2$ refers to the turbulence refraction structure parameter. For the application of VLC with the line-of-sight links, when the angle of incidence is within the receiver field-of-view, the the channel loss can be written as \cite{Kahn}
	\begin{equation}
		\zeta_\mathrm{VLC}=\frac{A_{\rm d}}{L^2} R_o(\upsilon_{\rm t}) G(\upsilon_{\rm r}) \mathrm{cos}(\upsilon_{\rm r}),
	\end{equation}
	where $A_{\rm d}$ refers to the physical area of the detector, $\upsilon_{\rm t}$ is the radiance angle, $\upsilon_{\rm r}$ is the angle of incidence at the receiver, $R_o(\upsilon_{\rm t})$ denotes the radiant intensity, and $G(\upsilon_{\rm r})$ is the concentrator gain.  }

According to the different employed quenching circuits, SPAD
can be classified into two main types, i.e., active quenching (AQ) and passive quenching (PQ) SPADs. AQ SPADs are normally with high complexity and cost, and are usually designed with small
array sizes \cite{Chitnis14}; whereas, PQ SPADs benefit from the simpler circuit design, higher PDE and low cost, therefore are commonly employed in the commercial low-cost photon counting receivers with large array sizes \cite{Matthews21}. In this work, we hence consider that the employed SPAD receiver is PQ-based. 
For each PQ SPAD in the array receiver, the average detected photon count during the time-domain OFDM sample duration $T_{\rm s}$ is given by\cite{Eisele11,HuangHy}
\begin{equation}\label{meanPQsingle}
	\mu_{\rm s}(x[n])= 	\frac{\lambda_\mathrm{a}[n]T_{\rm s}}{N_{\rm a}} \exp\left(-\frac{\lambda_\mathrm{a}[n]\tau_{\rm d}}{N_{\rm a}}\right), 
\end{equation}
where $N_{\rm a}$ denotes the number of SPAD pixels in the array and $\tau_{\rm d}$ refers to the dead time. The average detected photon count of the array receiver can be expressed as
\begin{equation}\label{meanPQ}
	\mu_{\rm a}(x[n])= N_{\rm a}\mu_{\rm s} (x[n])={\lambda_\mathrm{a}[n]T_{\rm s}} \exp\left(-\frac{\lambda_\mathrm{a}[n]\tau_{\rm d}}{N_{\rm a}}\right).
\end{equation}	
It is shown in (\ref{meanPQ}) that with the increase of the incident photon rate $\lambda_\mathrm{a}[n]$, the detected photon count firstly increases and then decreases. This indicates a nonlinear distortion to the received signal which results from the paralysis property of the PQ SPAD. The received photon rate which gives the highest detected photon count is $N_{\rm a}/\tau_{\rm d}$ and the corresponding detected photon count is $N_{\rm a}T_{\rm s}/(e \tau_{\rm d})$ where $e$ is the Euler's number \cite{Eisele11}. Note that in the absence of dead time, i.e., $\tau_{\rm d}=0$, the receiver nonlinearity vanishes, which is the case when an ideal photon counting receiver is considered. 

The variance of the detected photon count during $T_{\rm s}$ for a single PQ SPAD can be written as \cite{omote,Daniel00}
\begin{equation}
	\sigma_{\rm s}^2 (x[n])=\mu_{\rm s} (x[n])-\mu_{\rm s}^2 (x[n])\left[1-\left(1-\frac{\tau_{\rm d}}{T_{\rm s}}\right)^2\right],
\end{equation}
where it is assumed that the sample duration is larger than the dead time, i.e., $T_{\rm s}\geq \tau_{\rm d}$, which is the operation regime considered in this work. As the detected photon counts of SPADs in the array are independent, the variance of the detected photon count of the array is given by $\sigma^2_{\rm a}(x[n])=N_{\rm a}\sigma_{\rm s}^2 (x[n])$. Invoking the expression of $\mu_{\rm s} (x[n])$ in (\ref{meanPQsingle}), $\sigma_{\rm a}^2 (x[n])$ can be expressed as 
\begin{equation}\label{varPQ}
	\sigma_{\rm a}^2 (x[n])=\lambda_\mathrm{a}[n]T_{\rm s} \exp\left(-\frac{\lambda_\mathrm{a} [n]\tau_{\rm d}}{N_{\rm a}}\right)-\frac{\lambda_\mathrm{a}^2 [n]T_{\rm s}\tau_{\rm d}}{N_{\rm a}} \exp\left(-\frac{2\lambda_\mathrm{a} [n]\tau_{\rm d}}{N_{\rm a}}\right)\left(2-\frac{\tau_{\rm d}}{T_{\rm s}}\right)
\end{equation}
In (\ref{varPQ}), $\sigma_{\rm a}^2 (x[n])$ is the variance of shot noise of the detected photon count, which is a nonlinear function of the signal amplitude. Note that the dead time effect has made this signal dependent shot noise more complicated than the APD shot noise, whose power is proportional to the signal amplitude \cite{Safari15}. 

Denote the detected photon count during sample period $T_s$ after S/P mapping at the receiver as $y[n]$. When the size of the SPAD array is relatively large, due to the CLT,  $y[n]$ is approximately Gaussian distributed \cite{Khalighi17} with mean and variance given by (\ref{meanPQ}) and (\ref{varPQ}), respectively. Hence $y[n]$ can be rewritten as
\begin{equation}\label{yn1}
	y[n]=\mu_{\rm a}(x[n])+w_{\rm s}[n],
\end{equation} 
where $w_{\rm s}[n]$ represents the signal dependent shot noise which is Gaussian distributed with zero mean and signal dependent variance given by (\ref{varPQ}). As shown in Fig. \ref{system_figure}, the photon counting signal $y[n]$ is then converted back to the frequency-domain using the FFT operation given by
\begin{equation}
	Y[k]=\frac{1}{\sqrt{K}} \sum_{n=0}^{K-1} y[n] e^{-\frac{2\pi nkj}{K}}. 
\end{equation} 
After the single-tap equalization and parallel-to-serial (P/S) mapping, the received QAM signal can be achieved. The QAM demodulation can then be applied which results in the recovered bit steam. 

\section{Theoretical Analysis of SPAD-based OFDM}\label{theore} 
In the considered SPAD OFDM system, two nonlinear distortions exist. The first is the clipping-induced distortion due to the limited dynamic range of the light source as presented in (\ref{xclipped}). Such distortion also exists in standard OFDM systems with linear photodetectors \cite{Dimitrov12}.  The second is the additional unique SPAD-induced distortion given in (\ref{meanPQ}).  In this work, we combine these two nonlinear distortions and investigate the performance of the SPAD OFDM system under the combined nonlinear distortion. By substituting (\ref{xclipped}), (\ref{xt}), and (\ref{lama2}) into (\ref{meanPQ}), the combined nonlinear distortion of the transmitted signal $x[n]$ can be expressed as 
\begin{equation}\label{distortPQ}
	\mu_{\rm a}(x[n])=
	\begin{cases}
		\left(\psi_1\kappa_{\rm t}+\psi_2\right)	T_{\rm s} \exp\left[-\frac{\left(\psi_1\kappa_{\rm t}+\psi_2\right)	\tau_{\rm d}}{N_{\rm a}}\right], &  \mathrm{if}\,\,  x[n]\geq \kappa_{\rm t},\\[1.5ex]
		\left(\psi_1x[n]+\psi_2\right)T_{\rm s} \exp\left[-\frac{\left(\psi_1x[n]+\psi_2\right)	\tau_{\rm d}}{N_{\rm a}}\right],  &  \mathrm{if}\,\,  \kappa_{\rm b}<x[n]< \kappa_{\rm t},\\[1.5ex]
		{\left(\psi_1\kappa_{\rm b}+\psi_2\right)	T_{\rm s}} \exp\left[-\frac{\left(\psi_1\kappa_{\rm b}+\psi_2\right)	\tau_{\rm d}}{N_{\rm a}}\right], & \mathrm{if}\,\,  x[n]\leq \kappa_{\rm b}. 
	\end{cases}
\end{equation}
where 
\begin{equation}
	\psi_1=\mathcal{C}_{\rm s}\delta, \quad \rm{and}\,\,\, \psi_2=\mathcal{C}_{\rm s}P_{\rm bias}+\mathcal{C}_{\rm n}.
\end{equation}
The nonlinear distortion in an OFDM-based system can be described by a gain factor ($\alpha$) and an additional distortion-induced noise ($w_{\rm d}[n]$), both of which can be explained and quantified by the Bussgang theorem \cite{Tsonev13,Dimitrov12}. According to the Bussgang theorem, we have
\begin{equation}\label{BT1}
	\mu_{\rm a}(x[n])= \alpha x[n] + w_{\rm d}[n],
\end{equation}
where
\begin{equation}\label{BT2}	
{\mathbb{E}}\left\{x[n]w_{\rm d}[n]\right\} = 0,
\end{equation}
holds. Note that $\mathbb{E}\{\cdot\}$ denotes the expectation operation.  In the following discussion, both $\alpha$ and variance of $w_{\rm d}[n]$ will be derived.  
\begin{Proposition}\label{Prop_alpha}
The gain factor $\alpha$ can be expressed as 
\begin{align}\label{alpha_fin}
	\alpha=&\frac{\psi_1^2\tau_{\rm d}T_{\rm s}}{\sqrt{2\pi}N_{\rm a}}\left\{\exp \left[-\frac{\kappa_{\rm t}^2}{2}-\frac{\tau_{\rm d}}{N_{\rm a}}\left(\psi_1\kappa_{\rm t}+\psi_2\right)\right]
	-\exp\left[-\frac{\kappa_{\rm b}^2}{2}-\frac{\tau_{\rm d}}{N_{\rm a}}\left(\psi_1\kappa_{\rm b}+\psi_2\right)\right]\right\}\\[0.5em]
	&+{\psi_1}\,T_{\rm s}\,e^{-\frac{\psi_{2}\tau_{\rm d}}{N_{\rm a}}+\frac{\psi_{1}^2\tau_{\rm d}^2}{2N_{\rm a}^2}}\left[1+\frac{\psi_{1}^2\tau_{\rm d}^2}{N_{\rm a}^2}-\frac{\psi_{2}\tau_{\rm d}}{N_{\rm a}}\right]\left[Q\left(\kappa_{\rm b}+\frac{\psi_{1}\tau_{\rm d}}{N_{\rm a}}\right)-Q\left(\kappa_{\rm t}+\frac{\psi_{1}\tau_{\rm d}}{N_{\rm a}}\right)\right].\nonumber
\end{align} 
\end{Proposition}

\begin{proof}
Invoking (\ref{BT1}) and (\ref{BT2}), the gain factor $\alpha$ can be written as
\begin{equation}\label{alpha}
	\alpha = {{ \mathbb{E}}\left\{x\mu_{\rm a}(x)\right\}}/{{ \mathbb{E}}\left\{x^2\right\}}={{ \mathbb{E}}\left\{x\mu_{\rm a}(x)\right\}},
\end{equation}
where the sample index $n$ is dropped for simplicity. After some mathematical manipulations, one can get
\begin{equation}\label{alpha1}
\alpha=\frac{\left(\psi_{1} \kappa_{\rm t}+\psi_{2}\right)T_{\rm s}}{\sqrt{2\pi}}e^{-\frac{\left(\psi_{1} \kappa_{\rm t}+\psi_{2}\right) \tau_{\rm d}}{N_{\rm a}}-\frac{1}{2}\kappa_{\rm t}^2}-\frac{\left(\psi_{1} \kappa_{\rm b}+\psi_{2}\right)T_{\rm s}}{\sqrt{2\pi}}e^{-\frac{\left(\psi_{1} \kappa_{\rm b}+\psi_{2}\right) \tau_{\rm d}}{N_{\rm a}}-\frac{1}{2}\kappa_{\rm b}^2}+\mathcal{W},
\end{equation}
where the term $\mathcal{W}$ is given by 
\begin{equation}\label{Wx}
	\mathcal{W}=\int_{\kappa_{\rm b}}^{\kappa_{\rm t}} x\mu_{\rm a}(x) f(x) \,\mathrm{d} x.
\end{equation}
By substituting (\ref{distortPQ}) into (\ref{Wx}), $\mathcal{W}$ can be calculated as
\begin{align}\label{W}
	\mathcal{W}
	=&\frac{T_{\rm s}\psi_{1}}{\sqrt{2\pi}}\underbrace{\int_{\kappa_{\rm b}}^{\kappa_{\rm t}} x^2 {\rm exp}\left[-\frac{\left(\psi_{1} x+\psi_{2}\right)	\tau_{\rm d}}{N_{\rm a}}-\frac{1}{2}x^2\right] \,\mathrm{d} x}_{\mathcal{W}_1}\\
	&\qquad +\frac{T_{\rm s}\psi_{2}}{\sqrt{2\pi}}    \underbrace{\int_{\kappa_{\rm b}}^{\kappa_{\rm t}} x \, {\exp}\left[-\frac{\left(\psi_{1} x+\psi_{2}\right)	\tau_{\rm d}}{N_{\rm a}}-\frac{1}{2}x^2\right] \,\mathrm{d} x}_{\mathcal{W}_2}.\nonumber
\end{align} 
The integral $\mathcal{W}_1$ and $\mathcal{W}_2$ can be solved analytically as given by 
\begin{align}\label{W12}
\mathcal{W}_1=&\,e^{-\frac{\psi_{2} \tau_{\rm d}}{N_{\rm a}}+\frac{\psi_{1}^2\tau_{\rm d}^2}{2N_{\rm a}^2}}\Bigg\{\sqrt{\frac{\pi}{2}}\left(1\!+\!\frac{\psi_{1}^2\tau_{\rm d}^2}{N_{\rm a}^2}\right)\!\!\left[\mathrm{erf}\left(\frac{\kappa_{\rm t}+\frac{\psi_{1}\tau_{\rm d}}{N_{\rm a}}}{\sqrt{2}}\right)\!-\!\mathrm{erf}\!\left(\frac{\kappa_{\rm b}+\frac{\psi_{1}\tau_{\rm d}}{N_{\rm a}}}{\sqrt{2}}\right)\right]\!+\!g(\kappa_{\rm b})\!-\!g(\kappa_{\rm t})\Bigg\},\nonumber\\[0.5em]	\mathcal{W}_2=&\,e^{-\frac{\psi_{2}\tau_{\rm d}}{N_{\rm a}}+\frac{\psi_{1}^2\tau_{\rm d}^2}{2N_{\rm a}^2}}\Bigg[ \sqrt{2\pi}f\left(\kappa_{\rm b}+\frac{\psi_{1}\tau_{\rm d}}{N_{\rm a}}\right)-\sqrt{2\pi}f\left(\kappa_{\rm t}+\frac{\psi_{1}\tau_{\rm d}}{N_{\rm a}}\right)\nonumber\\[0.5em]
&\qquad \qquad +\sqrt{\frac{{\pi}}{2}}\frac{\psi_{1}\tau_{\rm d}}{N_{\rm a}} \mathrm{erfc}\left(\frac{\kappa_{\rm t}+\frac{\psi_{1}\tau_{\rm d}}{N_{\rm a}}}{\sqrt{2}}\right)-\sqrt{\frac{{\pi}}{2}}\frac{\psi_{1}\tau_{\rm d}}{N_{\rm a}} \mathrm{erfc}\left(\frac{\kappa_{\rm b}+\frac{\psi_{1}\tau_{\rm d}}{N_{\rm a}}}{\sqrt{2}}\right) \Bigg],
\end{align}
where
\begin{equation}
	g(\kappa)=\sqrt{2\pi}\left(\kappa-\frac{\psi_{1}\tau_{\rm d}}{N_{\rm a}}\right)f\left(\kappa+\frac{\psi_{1}\tau_{\rm d}}{N_{\rm a}}\right),
\end{equation}
and  $\mathrm{erfc}(\cdot)$ denotes the complementary error function. By substituting (\ref{W12}) and (\ref{W}) into (\ref{alpha1}), the analytical expression of $\alpha$ can be obtained as presented in (\ref{alpha_fin}).  
\end{proof}

We can consider a special case when $\tau_{\rm d}=0$ which refers to the linear ideal photon counting receiver. Substituting $\tau_{\rm d}=0$ into (\ref{alpha_fin}) leads to a gain factor of
\begin{equation}
	\alpha_{\rm id}=\psi_{1}T_{\rm s} \left[Q\left(\kappa_{\rm b}\right)-Q\left(\kappa_{\rm t}\right)\right].
\end{equation}      
This result is in line with the derived gain factor of OFDM systems with traditional linear receivers \cite{Dimitrov12,Chen16}.

According to (\ref{BT1}) and (\ref{BT2}), the variance of the distortion induced noise, denoted as $\sigma_{w_{\rm d}}^2$, can be calculated by
\begin{equation}\label{sigma_wd}
	\sigma_{w_{\rm d}}^2=\mathbb{E}\left\{\mu^2_{\rm a}(x)\right\}-\mathbb{E}^2\left\{\mu_{\rm a}(x)\right\}-\alpha^2.
\end{equation}
Therefore, to calculate $\sigma_{w_{\rm d}}^2$ the first two moments of $\mu_{\rm a}(x)$ should be investigated.
\begin{Proposition} 
The average of the received distorted signal $\mathbb{E}\left\{\mu_{\rm a}(x)\right\}$ can be expressed as 
	\begin{align}\label{mean_fin}
		\mathbb{E}\left\{\mu_{\rm a}(x)\right\}=&\left(\psi_{1} \kappa_{\rm b}+\psi_{2}\right)T_{\rm s}\,e^{-\frac{\left(\psi_{1} \kappa_{\rm b}+\psi_{2}\right) \tau_{\rm d}}{N_{\rm a}}}{Q}\left(-\kappa_{\rm b}\right)+\left(\psi_{1} \kappa_{\rm t}+\psi_{2}\right)T_{\rm s}\,e^{-\frac{\left(\psi_{1} \kappa_{\rm t}+\psi_{2}\right) \tau_{\rm d}}{N_{\rm a}}}{Q}\left(\kappa_{\rm t}\right)\nonumber\\[0.4em]
		&+\frac{T_{\rm s}\psi_{1}}{\sqrt{2\pi}}\,e^{-\frac{\psi_{2}\tau_{\rm d}}{N_{\rm a}}}\left[e^{-\frac{1}{2}\kappa_{\rm b}^2-\frac{\psi_{1}\tau_{\rm d}\kappa_{\rm b}}{N_{\rm a}}  }-e^{-\frac{1}{2}\kappa_{\rm t}^2-\frac{\psi_{1}\tau_{\rm d}\kappa_{\rm t}}{N_{\rm a}}  }\right]\nonumber\\[0.4em]
		&+{T_{\rm s}}\left(\frac{\tau_{\rm d}\psi_1^2}{N_a}-\psi_2\right)e^{-\frac{\psi_2\tau_{\rm d}}{N_{\rm a}}+\frac{\psi_1^2\tau_{\rm d}^2}{2N_{\rm a}^2}}\left[{Q}\left({\kappa_{\rm t}+\frac{\psi_{1}\tau_{\rm d}}{N_{\rm a}}}\right)-{Q}\left({\kappa_{\rm b}+\frac{\psi_{1}\tau_{\rm d}}{N_{\rm a}}}\right)\right].
	\end{align} 
\end{Proposition}

\begin{proof}
The expectation of $\mu_{\rm a}(x)$ is given by $\int_{-\infty}^{+\infty} \mu_{\rm a}(x) f(x) \,\mathrm{d} x$ which can be further calculated as 
\begin{align}\label{Eua}
	\mathbb{E}\left\{\mu_{\rm a}(x)\right\}=&\left(\psi_{1} \kappa_{\rm b}+\psi_{2}\right)T_{\rm s}e^{-\frac{\left(\psi_{1} \kappa_{\rm b}+\psi_{2}\right) \tau_{\rm d}}{N_{\rm a}}}Q\left(-\kappa_{\rm b}\right)\\
	&\,+\left(\psi_{1} \kappa_{\rm t}+\psi_{2}\right)T_{\rm s}e^{-\frac{\left(\psi_{1} \kappa_{\rm t}+\psi_{2}\right) \tau_{\rm d}}{N_{\rm a}}}Q\left(\kappa_{\rm t}\right)+\underbrace{\int_{\kappa_{\rm b}}^{\kappa_{\rm t}} \mu_{\rm a}(x) f(x) \,\mathrm{d} x}_{\mathcal{M}}.\nonumber
\end{align}
After some manipulations, the integral $\mathcal{M}$ can be expressed as 
\begin{equation}\label{M}
	\mathcal{M}=\frac{T_{\rm s}\psi_{1}}{\sqrt{2\pi}}\mathcal{W}_2+\frac{T_{\rm s}\psi_{2}}{\sqrt{2\pi}}e^{-\frac{\psi_{2}\tau_{\rm d}}{N_{\rm a}}}\underbrace{\int_{\kappa_{\rm b}}^{\kappa_{\rm t}} e^{-\frac{1}{2}x^2-\frac{\psi_{1}\tau_{\rm d}}{N_{\rm a}}x} \,\mathrm{d} x}_{\mathcal{M}_1}.
\end{equation}
The term $\mathcal{M}_1$ can be solved as 
\begin{equation}\label{M1}
	\mathcal{M}_1=\sqrt{\frac{\pi}{2}}\,e^{\frac{\psi_{1}^2\tau_{\rm d}^2}{2N_{\rm a}^2}}\left[\mathrm{erf}\left(\frac{\kappa_{\rm t}+\frac{\psi_{1}\tau_{\rm d}}{N_{\rm a}}}{\sqrt{2}}\right)-\mathrm{erf}\left(\frac{\kappa_{\rm b}+\frac{\psi_{1}\tau_{\rm d}}{N_{\rm a}}}{\sqrt{2}}\right)\right].
\end{equation}
Therefore, by substituting (\ref{M1}), (\ref{M}) into (\ref{Eua}) and after some mathematical manipulations, the first moment of the received distorted signal can be written as (\ref{mean_fin}).
\end{proof}

\begin{Proposition} 
The second moment of the received distorted signal, i.e.,  $\mathbb{E}\left\{\mu^2_{\rm a}(x)\right\}$ is calculated as 
	\begin{align}\label{sec_mom}
	\mathbb{E}&\left\{\mu_{\rm a}^2(x)\right\}=\\
	&\;\;\frac{T_{\rm s}^2}{\sqrt{2\pi}}e^{-\frac{2\left(\psi_{1} \kappa_{\rm b}+\psi_{2}\right) \tau_{\rm d}}{N_{\rm a}}}\bigg\{ \left[\psi_1^2\kappa_{\rm b}-\frac{2\psi_1^3\tau_{\rm d}}{N_{\rm a}}+2\psi_1\psi_2\right]e^{-\frac{1}{2}\kappa_{\rm b}^2}+  \sqrt{2\pi}\left(\psi_{1} \kappa_{\rm b}+\psi_{2}\right)^2Q(-\kappa_{\rm b})\bigg\}\nonumber\\[0.5em]
	&\;\;\;\;-\frac{T_{\rm s}^2}{\sqrt{2\pi}}e^{-\frac{2\left(\psi_{1} \kappa_{\rm t}+\psi_{2}\right) \tau_{\rm d}}{N_{\rm a}}}\bigg\{ \left[\psi_1^2\kappa_{\rm t}-\frac{2\psi_1^3\tau_{\rm d}}{N_{\rm a}}+2\psi_1\psi_2\right]e^{-\frac{1}{2}\kappa_{\rm t}^2}- \sqrt{2\pi}\left(\psi_{1} \kappa_{\rm t}+\psi_{2}\right)^2Q(\kappa_{\rm t})\bigg\}\nonumber\\[0.5em]
	&\;\;\;\;+{T_{\rm s}^2}e^{-\frac{2\psi_2\tau_{\rm d}}{N_{\rm a}}+\frac{2\psi_1^2\tau_{\rm d}^2}{N_{\rm a}^2}}\left[\psi_1^2+\left(\frac{2\psi^2_1\tau_{\rm d}}{N_{\rm a}}-{\psi_2}\right)^2\right]\left[Q\bigg(\kappa_{\rm b}+\frac{2\psi_{1}\tau_{\rm d}}{N_{\rm a}}\bigg)-Q\bigg(\kappa_{\rm t}+\frac{2\psi_{1}\tau_{\rm d}}{N_{\rm a}}\bigg)\right].\nonumber
\end{align}
\end{Proposition} 

\begin{proof}
The second moment of $\mu_{\rm a}(x)$ is given by $\mathbb{E}\left\{\mu^2_{\rm a}(x)\right\}=\int_{-\infty}^{+\infty} \mu^2_{\rm a}(x) f(x) \,\mathrm{d} x$ which can be further calculated as 
\begin{align}\label{Eua2}
\mathbb{E}\left\{\mu^2_{\rm a}(x)\right\}=&{\left(\psi_{1} \kappa_{\rm t}+\psi_{2}\right)^2T_{\rm s}^2}e^{-\frac{2\left(\psi_{1} \kappa_{\rm t}+\psi_{2}\right) \tau_{\rm d}}{N_{\rm a}}}Q(\kappa_{\rm t})
\\
&\,\,+{\left(\psi_{1} \kappa_{\rm b}+\psi_{2}\right)^2T_{\rm s}^2}e^{-\frac{2\left(\psi_{1} \kappa_{\rm b}+\psi_{2}\right) \tau_{\rm d}}{N_{\rm a}}}Q(-\kappa_{\rm b})+\mathcal{U},	\nonumber
\end{align}
where the term $\mathcal{U}=\int_{\kappa_{\rm b}}^{\kappa_{\rm t}} \mu^2_{\rm a}(x) f(x) \,\mathrm{d} x$. After some mathematical calculations, $\mathcal{U}$ can be written as
\begin{equation}\label{U}
\mathcal{U}=\frac{T_{\rm s}^2}{\sqrt{2\pi}}\,e^{-\frac{2\psi_2\tau_{\rm d}}{N_{\rm a}}+\frac{2\psi_1^2\tau_{\rm d}^2}{N_{\rm a}^2}}\left( \psi_{1}^2\,{\mathcal{U}_1}+2\psi_1\psi_2\,\mathcal{U}_2+\psi_2^2\,\mathcal{U}_3\right),
\end{equation}
where 
\begin{equation}
\mathcal{U}_1=\int_{\kappa_{\rm b}}^{\kappa_{\rm t}}{x^2}e^{-\frac{1}{2}\left(x+\frac{2\psi_{1}\tau_{\rm d}}{N_{\rm a}}\right)^2}\,\mathrm{d} x=e^{\frac{\psi_{2}\tau_{\rm d}}{N_{\rm a}}-\frac{2\psi_{1}^2\tau_{\rm d}^2}{N_{\rm a}^2}}\mathcal{W}_1',
\end{equation}
\begin{equation}
	\mathcal{U}_2=\int_{\kappa_{\rm b}}^{\kappa_{\rm t}}xe^{-\frac{1}{2}\left(x+\frac{2\psi_{1}\tau_{\rm d}}{N_{\rm a}}\right)^2}\,\mathrm{d} x=e^{\frac{\psi_{2}\tau_{\rm d}}{N_{\rm a}}-\frac{2\psi_{1}^2\tau_{\rm d}^2}{N_{\rm a}^2}}\mathcal{W}_2',
\end{equation}
and 
\begin{equation}
	\mathcal{U}_3=\int_{\kappa_{\rm b}}^{\kappa_{\rm t}}e^{-\frac{1}{2}\left(x+\frac{2\psi_{1}\tau_{\rm d}}{N_{\rm a}}\right)^2}\,\mathrm{d} x=\sqrt{{2\pi}}\left[Q\left(\kappa_{\rm b}+\frac{2\psi_{1}\tau_{\rm d}}{N_{\rm a}}\right)-Q\left(\kappa_{\rm t}+\frac{2\psi_{1}\tau_{\rm d}}{N_{\rm a}}\right)\right].
\end{equation}	
Note that $\mathcal{W}_1'$ and $\mathcal{W}_2'$  are the same as $\mathcal{W}_1$ and $\mathcal{W}_2$ presented in (\ref{W12}) but replacing $\psi_{1}$ by $2\psi_{1}$. Finally, by plugging (\ref{U}) into (\ref{Eua2}), the second moment of $\mu_{\rm a}(x)$ can be calculated as presented in (\ref{sec_mom}). 
\end{proof}

By substituting (\ref{alpha_fin}), (\ref{mean_fin}), and (\ref{sec_mom}) into (\ref{sigma_wd}), the analytical expression for the variance of the distortion-induced noise $\sigma_{w_{\rm d}}^2$ can be achieved. For ideal photon counting receiver,  by substituting $\tau_{\rm d}=0$ into (\ref{sigma_wd}), this distortion-induced noise turns to be only introduced by the signal clipping as 
\begin{align}
	\sigma_{w_{\rm d, id}}^2=\;&\psi_{1}^2T_{\rm s}^2\left\{Q(\kappa_{\rm b})-Q(\kappa_{\rm t})+\kappa_{\rm b}f(\kappa_{\rm b})-\kappa_{\rm t}f(\kappa_{\rm t}) \right\}\nonumber\\[2pt]
	&+\psi_{1}^2T_{\rm s}^2\left\{\kappa_{\rm b}^2Q(-\kappa_{\rm b})+\kappa_{\rm t}^2Q(\kappa_{\rm t})-\left[Q(\kappa_{\rm b})-Q(\kappa_{\rm t})\right]^2    \right\}\nonumber\\[2pt]
	&-\psi_{1}^2T_{\rm s}^2\left[f(\kappa_{\rm b})-f(\kappa_{\rm t})+\kappa_{\rm b}Q(-\kappa_{\rm b})+\kappa_{\rm t}Q(\kappa_{\rm t})\right]^2. 
\end{align}  
This result is again in line with the derived clipping noise when linear receiver is employed \cite{Dimitrov12,Chen16}.

By plugging (\ref{BT1}) into (\ref{yn1}), the SPAD output $y[n]$ can be rewritten as
\begin{equation}\label{yn2}
	y[n]=\alpha x[n]+w_{\rm d}[n]+w_{\rm s}[n].
\end{equation}
The analytical expressions of both $\alpha$ and $\sigma_{w_{\rm d}}^2$ have been derived as shown in Proposition $1$ to $3$. After applying FFT operation, the received time-domain signal $y[n]$ is converted back to the frequency domain which can be expressed as
\begin{equation}\label{yn3}
	Y[k]=\alpha X[k]+W_{\rm d}[k]+W_{\rm s}[k],
\end{equation}
where $W_{\rm d}[k]$ and $W_{\rm s}[k]$ denote the FFT of $w_{\rm d}[n]$ and $w_{\rm s}[n]$, respectively. According to CLT, both $W_{\rm d}[k]$ and $W_{\rm s}[k]$ are zero-mean Gaussian distributed noise terms when the number of subcarriers is sufficiently large \cite{Dimitrov13,Tsonev13}. 
Therefore, the received signal in the frequency domain is given by the transmitted signal multiplied by a gain factor plus additive Gaussian noise induced by nonlinear distortion and shot noise. 
For the distortion-induced noise $w_{\rm d}[n]$, since it is uncorrelated with the signal according to the Bussgang theorem, the frequency-domain noise is also uncorrelated with the signal thus its variance is equal to the time-domain noise variance \cite{Tsonev13}. Therefore, one has $\sigma_{W_{\rm d}}^2=\sigma_{w_{\rm d}}^2$ where $\sigma_{W_{\rm d}}^2$ denotes the variance of $W_{\rm d}[k]$. Different from the distortion-induced noise, the shot noise $w_{\rm s}[n]$ is signal dependent. In the following proposition, we  derive the shot noise variance in frequency domain, denoted as $\sigma_{W_{\rm s}}^2$.  

\begin{Proposition} \label{shotnoise}
The analytical expression of the frequency-domain shot noise variance $\sigma_{W_{\rm s}}^2$ is given by
	\begin{align}\label{sigws}
	\sigma&_{W_{\rm s}}^2=\sigma_{\rm a}^2(\kappa_{\rm b})Q(-\kappa_{\rm b})+\sigma_{\rm a}^2(\kappa_{\rm t})Q(\kappa_{\rm t})+\frac{T_{\rm s}\psi_{1}}{\sqrt{2\pi}}\,e^{-\frac{\psi_{2}\tau_{\rm d}}{N_{\rm a}}}\left[e^{-\frac{1}{2}\kappa_{\rm b}^2-\frac{\psi_{1}\tau_{\rm d}\kappa_{\rm b}}{N_{\rm a}}  }-e^{-\frac{1}{2}\kappa_{\rm t}^2-\frac{\psi_{1}\tau_{\rm d}\kappa_{\rm t}}{N_{\rm a}}  }\right]\nonumber\\[0.5em]
	&+{T_{\rm s}}\!\left(\!\frac{\tau_{\rm d}\psi_1^2}{N_a}\!-\!\psi_2\right)\!e^{-\frac{\psi_2\tau_{\rm d}}{N_{\rm a}}+\frac{\psi_1^2\tau_{\rm d}^2}{2N_{\rm a}^2}}\Bigg[Q\bigg(\kappa_{\rm t}\!+\!\frac{\psi_{1}\tau_{\rm d}}{N_{\rm a}}\bigg)\!-\!Q\bigg(\kappa_{\rm b}\!+\!\frac{\psi_{1}\tau_{\rm d}}{N_{\rm a}}\bigg)\Bigg]\!-\!\frac{\left(2T_{\rm s}\!-\!{\tau_{\rm d}}\right)\tau_{\rm d}}{\sqrt{2\pi}N_{\rm a}}e^{-\frac{2\psi_2\tau_{\rm d}}{N_{\rm a}}+\frac{2\psi_1^2\tau_{\rm d}^2}{N_{\rm a}^2}}\nonumber\\[0.5em]
	&\times 
	\Bigg\{\left[\psi_1^2\kappa_{\rm b}-\frac{2\psi_1^3\tau_{\rm d}}{N_{\rm a}}+2\psi_1\psi_2\right]e^{-\frac{1}{2}\left(\kappa_{\rm b}+\frac{2\psi_1\tau_{\rm d}}{N_{\rm a}}\right)^2} 
	-\left[\psi_1^2\kappa_{\rm t}-\frac{2\psi_1^3\tau_{\rm d}}{N_{\rm a}}+2\psi_1\psi_2\right]e^{-\frac{1}{2}\left(\kappa_{\rm t}+\frac{2\psi_1\tau_{\rm d}}{N_{\rm a}}\right)^2} \nonumber\\[0.5em]
	&+ \bigg[\sqrt{{2\pi}}\psi_1^2+\sqrt{2\pi}\left(\frac{2\psi^2_1\tau_{\rm d}}{N_{\rm a}}-{\psi_2}\right)^2\bigg]\bigg[Q\bigg(\kappa_{\rm b}+\frac{2\psi_{1}\tau_{\rm d}}{N_{\rm a}}\bigg)-Q\bigg(\kappa_{\rm t}+\frac{2\psi_{1}\tau_{\rm d}}{N_{\rm a}}\bigg)\bigg]\Bigg\}.
\end{align}

\end{Proposition} 

\begin{proof}
The variance $\sigma_{W_{\rm s}}^2$ can be expressed as
\begin{align}\label{Ws1}
    \sigma_{W_{\rm s}}^2&=\mathbb{E}\{|W_{\rm s}[k]|^2\}\\[2pt]
    &=\frac{1}{K}\sum_{n=0}^{K-1}\sum_{m=0}^{K-1}\mathbb{E}\{ w_s[n]w_s^*[m]\}e^{-\frac{2\pi nkj}{K}+\frac{2\pi mkj}{K}}.\nonumber\\[2pt]
    &=\frac{1}{K}\sum_{n=0}^{K-1}\mathbb{E}\{ |w_s[n]|^2\}-\mathbb{E}^2\{w_s[n]\}.
\end{align}
Considering that $w_s[n]$ is with zero mean, (\ref{Ws1}) can be simplified as
\begin{equation}\label{sigmaWs2}
    \sigma_{W_{\rm s}}^2 =\frac{1}{K} \sum_{n=0}^{K-1} \sigma_{\rm a}^2(x[n]) \simeq \mathbb{E}\{\sigma_{\rm a}^2(x[n])\},
\end{equation}
where the approximation is accurate for relatively large FFT size $K$. Equation (\ref{sigmaWs2}) suggests that the variance of the shot noise in the frequency domain is equal to the average of the signal dependent shot noise variance in the time domain, and hence is signal independent.  The average of the shot noise in time domain can be also written as 
\begin{equation}\label{sigmaws}
	\mathbb{E}\{\sigma_{\rm a}^2(x)\}=\int_{-\infty}^{+\infty}f(x) \sigma_{\rm a}^2(x)\, \mathrm{d}x,
\end{equation}
where again the sample index $n$ is dropped for simplicity and 
$f(x)$ is defined in (\ref{fx}). Note that the expression of $\sigma_{\rm a}^2(x)$ can be achieved by substituting (\ref{xclipped}), (\ref{xt}) and (\ref{lama2}) into (\ref{varPQ}).   

Equation (\ref{sigmaws}) can be solved as 
\begin{equation}\label{sigws2}
	\sigma_{W_{\rm s}}^2=\sigma_{\rm a}^2(\kappa_{\rm b})Q(-\kappa_{\rm b})+\sigma_{\rm a}^2(\kappa_{\rm t})Q(\kappa_{\rm t})+\!\!\underbrace{\int_{\kappa_{\rm b}}^{\kappa_{\rm t}}f(x) \sigma_{\rm a}^2(x)\, \mathrm{d}x}_{\mathcal{S}}.
\end{equation}
Now let's derive the analytical expression of the term $\mathcal{S}$. According to (\ref{varPQ}), the conditional variance $\sigma_{\rm a}^2(x)$ when $x\in[\kappa_{\rm b},\kappa_{\rm t}]$ is given by
\begin{equation}\label{siga}
	\sigma_{\rm a}^2(x)=\left(\psi_{1}x+\psi_{2}\right)T_{\rm s}\,e^{-\frac{\left(\psi_1x+\psi_2\right)\tau_{\rm d}}{N_{\rm a}}}-\frac{\left(\psi_{1}x+\psi_{2}\right)^2T_{\rm s}\tau_{\rm d}}{N_{\rm a}}\left(2-\frac{\tau_{\rm d}}{T_{\rm s}}\right)e^{-\frac{2\left(\psi_1x+\psi_2\right)\tau_{\rm d}}{N_{\rm a}}}.
\end{equation}
Substituting (\ref{siga}) into $\mathcal{S}$ results in
\begin{align}\label{S}
	\mathcal{S}=&\underbrace{\int_{\kappa_{\rm b}}^{\kappa_{\rm t}}\left(\psi_{1}x+\psi_{2}\right)T_{\rm s}e^{-\frac{\left(\psi_1x+\psi_2\right)\tau_{\rm d}}{N_{\rm a}}}f(x)\, \mathrm{d}x}_{\	\mathcal{S}_1}\\[2pt]
	&\;\;\;-\underbrace{\int_{\kappa_{\rm b}}^{\kappa_{\rm t}}\frac{\left(\psi_{1}x+\psi_{2}\right)^2T_{\rm s}\tau_{\rm d}}{N_{\rm a}}\left(2-\frac{\tau_{\rm d}}{T_{\rm s}}\right) e^{-\frac{2\left(\psi_1x+\psi_2\right)\tau_{\rm d}}{N_{\rm a}}}f(x)\, \mathrm{d}x}_{\mathcal{S}_2}.\nonumber
\end{align}
The first integral in (\ref{S}) can be solved analytically as 
\begin{equation}\label{S1}
	\mathcal{S}_1=\frac{\psi_1 T_{\rm s}}{\sqrt{2\pi}}\,\mathcal{W}_2+\psi_{2}\,T_{\rm s}\,e^{-\frac{\psi_{2}\tau_{\rm d}}{N_{\rm a}}+\frac{\psi_1^2\tau_{\rm d}^2}{2N_{\rm a}^2}}\left[Q\left(\kappa_{\rm b}+\frac{\psi_{1}\tau_{\rm d}}{N_{\rm a}}\right)-Q\left(\kappa_{\rm t}+\frac{\psi_{1}\tau_{\rm d}}{N_{\rm a}}\right)\right].
\end{equation}
The second integral can be calculated as 
\begin{equation}\label{S2}
	\mathcal{S}_2=\left(2-\frac{\tau_{\rm d}}{T_s}\right)\frac{\tau_{\rm d}}{N_{\rm a}T_{\rm s}}\mathcal{U},
\end{equation} 
where the term $\mathcal{U}$ is given by (\ref{U}). Therefore, by substituting (\ref{S1}), (\ref{S2}) and (\ref{S}) into (\ref{sigws2}), the expression of $\sigma_{W_{\rm s}}^2$ can be achieved which is summarized as (\ref{sigws}).  
\end{proof}

Again, for the special case of ideal photon counting receiver, the variance of the shot noise in the frequency domain can be simplified as
\begin{equation}
	\sigma_{W_{\rm s, id}}^2=\left(\psi_{1}\kappa_{\rm b}+\psi_{2}\right)T_{\rm s}-\psi_{1}\kappa_{\rm b}Q(\kappa_{\rm b})T_{\rm s}+\psi_{1}\kappa_{\rm t}Q(\kappa_{\rm t})T_{\rm s}+\psi_{1}T_{\rm s}f(\kappa_{\rm b})-\psi_{1}T_{\rm s}f(\kappa_{\rm t}). 
\end{equation} 

Proposition \ref{shotnoise} implies that for OFDM transmission, due to the averaging effect of the FFT operation, the time-domain signal-dependent shot noise turns to be signal-independent noise in the frequency domain. Therefore, both noise terms in the received signal shown in (\ref{yn3}) are uncorrelated with the signal making it a standard additive Gaussian noise channel model. Note that unlike the original OWC channel model, the equivalent additive Gaussian channel model for OFDM SPAD-based OWC system in (\ref{yn3}) is not constrained by the non-negativity condition of the intensity modulation. Therefore, considering symbol-by-symbol detection, the achievable SE of the system can be upper bounded by
\begin{equation}\label{SEupper}
    \varsigma_{\rm upper}=\log_2(1+\gamma).
\end{equation}
where the SNR of the received signal can be written as
\begin{equation}\label{gamma}
	\gamma=\frac{\alpha^2\frac{K}{K-2}}{\sigma_{W_{\rm d}}^2+\sigma_{W_{\rm s}}^2}=\frac{1}{\frac{1}{\gamma_{\rm d}}+\frac{1}{\gamma_{\rm s}}}.
\end{equation}
The terms 
\begin{equation}
	\gamma_{\rm d}=\frac{\alpha^2\frac{K}{K-2}}{\sigma_{W_{\rm d}}^2}, \quad \mathrm{and} \quad \gamma_{\rm s}=\frac{\alpha^2\frac{K}{K-2}}{\sigma_{W_{\rm s}}^2},
\end{equation}
denote the signal-to-distortion-noise ratio (SDNR) and signal-to-shot-noise ratio (SSNR), respectively. The derivation of the analytical expression of the SNR is now complete. The term $\alpha^2{K}/{(K-2)}$ refers to the received electrical signal power, therefore the absolute value of the gain factor, i.e., $|\alpha|$, can be used to measure the signal power. In addition, the value of SNR is determined by two factors, i.e., SDNR and SSNR. The BER of the system when $M$-QAM is employed can be further written as \cite{Dimitrov12}
\begin{equation}\label{BER}
	P_e=\frac{4\big(\sqrt{M}-1\big)}{\sqrt{M}\mathrm{log}_2(M)}{Q} \Bigg(\sqrt{\frac{3\gamma}{M-1}}\Bigg)+\frac{4\big(\sqrt{M}-2\big)}{\sqrt{M}\mathrm{log}_2(M)}{Q} \Bigg(3\sqrt{\frac{3\gamma}{M-1}}\Bigg).
\end{equation}

\section{Numerical Results}\label{NumRes}
\begin{table}
	\renewcommand{\arraystretch}{1.1}
	\caption{The Parameter Setting  \cite{ON,Patanwala22}}
	\label{table}
	\centering
	\resizebox{0.5\textwidth}{!}
	{\begin{tabular}{|c|c|c|}
			\hline
			Symbol & Definition & Value\\
			\hline\hline
			$\lambda_\mathrm{op}$ & Optical wavelength & $450$ nm\\
			\hline
			$\Upsilon_{\rm PDE}$ & The PDE of SPAD @ $\lambda_\mathrm{op}$ & $0.35$ \\ 
			\hline 
			$N_{\rm a}$ &  Number of SPAD pixels in the array & $8192$ \\
			\hline
			$\tau_{\rm d}$ & Dead time of SPAD &  $10$ ns\\
			\hline
			$P_{\rm B}$ & Background light power & $10$ nW\\
			\hline
			$\vartheta_{\rm DCR}$ & Dark count rate & $0.5$ MHz\\
			\hline
			$\varphi_{\rm AP}$ & Afterpulsing probability &  $0.75\%$\\
			\hline
			$\varphi_{\rm CT}$ & Crosstalk probability & $2.5\%$\\
			\hline 
			$P_{\rm max}$ & Maximal transmitted power & $10$ mW\\
			\hline 
			$K$ & Size of FFT and IFFT & $1024$ \\
			\hline
			$T_{\rm s}$ & Time-Domain sample duration  & $20$ ns\\
			\hline
			$M$  & Modulation order & $[16,\,32]$ \\
			\hline
	\end{tabular}}
\end{table}
In this section, the numerical performance analysis of the SPAD-based OFDM system is presented. Unless otherwise mentioned, the parameters used in the simulation
are given in Table \ref{table}. For simplicity, we assume that symmetric clipping is employed, i.e., $\kappa_{\rm t}=-\kappa_{\rm b}=\kappa$, and $P_{\rm min}=0$, hence the average transmitted optical power is $\overline{P}_{\rm Tx}=P_{\rm bias}=P_{\rm max}/2$. By changing channel path loss $\zeta$,  various average optical power at the receiver $\overline{P}_{\rm Rx}=\zeta\,\overline{P}_{\rm Tx}$ can be achieved. The proposed framework in this work is for the  performance analysis of the general OFDM-based OWC systems. Therefore, in this section, we focus on the performance of the systems under various received power. Note that for the specific OWC applications, e.g., FSO and VLC, disparate channel effects should be further considered to achieve the average or outage performance. 

\begin{figure}[!t]
	\centering\includegraphics[width=0.7\textwidth]{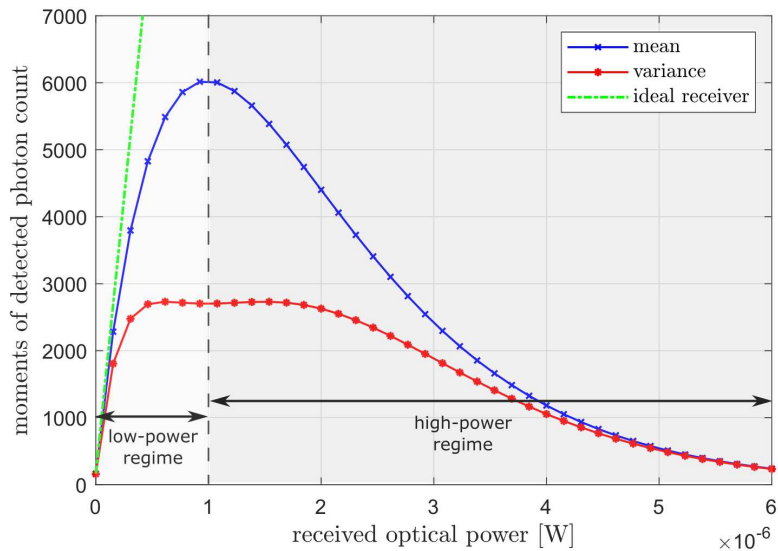}
	\caption{The mean and variance of the detected photon count of the considered SPAD receiver versus the received optical signal power. }\label{moment_SPAD_output}
\end{figure}
Fig. \ref{moment_SPAD_output} presents the mean and variance of the SPAD receiver output versus the received optical power, which are calculated based on (\ref{meanPQ}) and (\ref{varPQ}). With the increase of the received optical power, the average value of the detected photon count firstly increases and then decreases due to the paralysis characteristics of the employed PQ SPAD. The maximal detected photon count of the SPAD receiver is given by $N_{\rm a}T_{\rm s}/(e \tau_{\rm d})$ which is around $6000$ in the considered system. For the considered receiver,  this photon count is achieved when $\overline{P}_{\rm Rx}$ is around $1$ $\mu$W. In the following discussion, we denote this power as the saturation power of the receiver. The regimes with power less and higher than this power are denoted as low- and high-power regime, respectively. In addition, the variance of the photon count is signal dependent and is less than the mean value, which shows the sub-Poisson property of the SPAD receiver \cite{Huang22}. This is different from the ideal photon counting receiver, i.e., when $\tau_{\rm d}=0$, whose mean and variance of the detected photon count are identical and both increase linearly with the received optical power, as also plotted in Fig. \ref{moment_SPAD_output}. Therefore, the dead time induced nonlinear effect can strongly reduce the detected photon count and hence limits the communication performance.            

\begin{figure}[!t]
	\centering\includegraphics[width=0.7\textwidth]{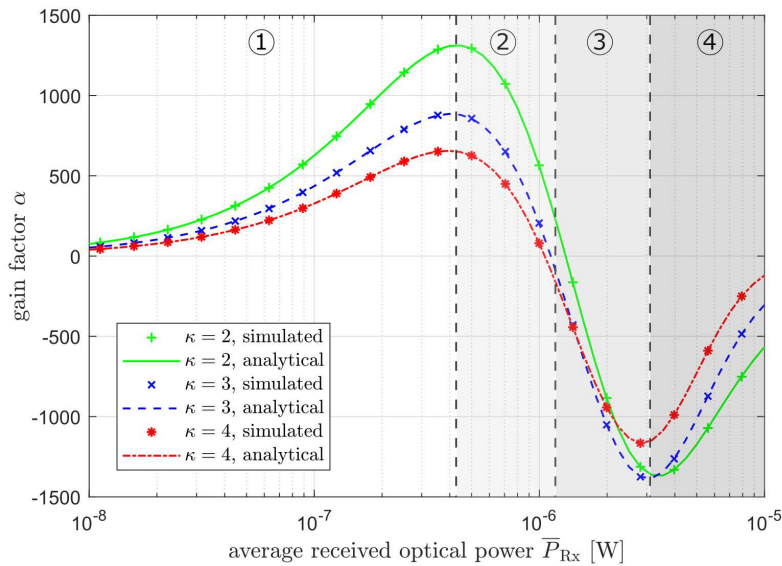}
	\caption{The gain factor $\alpha$ versus the received optical power with various clipping level. }\label{alpha_gain}
\end{figure}
The gain factor $\alpha$ versus the $\overline{P}_{\rm Rx}$ with various $\kappa$ is demonstrated in Fig. \ref{alpha_gain}. It is shown that the analytical results calculated based on (\ref{alpha_fin}) exactly match with the simulation ones, which justifies our analytical derivations. Note that as mentioned above, $\alpha^2K/(K-2)$ refers to the received electrical signal power, hence larger $|\alpha|$ indicates higher electrical signal power. It is presented that initially in Regime $1$, with the rise of $\overline{P}_{\rm Rx}$, $\alpha$ increases but then drops in Regime $2$ because of the severer SPAD nonlinearity. In particular, when $\overline{P}_{\rm Rx}$ is around the saturation power, $\alpha$ approaches zero, since at this point the received optical waveform experiences significant folding effect distorting the modulated signal. This means that this received power would result in the least undistorted electrical signal power and hence the worst performance. Further increase of $\overline{P}_{\rm Rx}$ leads to a negative $\alpha$ in Regime $3$. This is because when $\overline{P}_{\rm Rx}$ is beyond the saturation power, SPAD receiver operates in the high-power regime where any increase of the optical power would cause the reduction of the detected photon count, as presented in Fig. \ref{moment_SPAD_output}. Invoking the definition of $\alpha$ in (\ref{alpha}), such negative correlation between received power and detected photon count introduces a negative $\alpha$. Note that the increase of $|\alpha|$ with the rise of $\overline{P}_{\rm Rx}$ can be observed in Regime $3$ again. This is because when the average signal power goes to this regime, the majority of the waveform dynamic range moves to a linear part of the SPAD receiver response where the detected photon count monotonically decreases with $\overline{P}_{\rm Rx}$.  However, when $\overline{P}_{\rm Rx}$ further increase to the levels of Regime $4$, $\alpha$ starts to tend to zero again. In this scenario, the slope reduction of the average photon count versus the received power (see Fig. \ref{moment_SPAD_output}) shrinks the dynamic range of the detected photon count signal and hence reduces the electrical signal power. In addition, it is also demonstrated that $\alpha$ also varies with the clipping level $\kappa$ and generally a lower $\kappa$ can provide a higher electrical signal power.       

\begin{figure}[!t]
	\centering\includegraphics[width=0.7\textwidth]{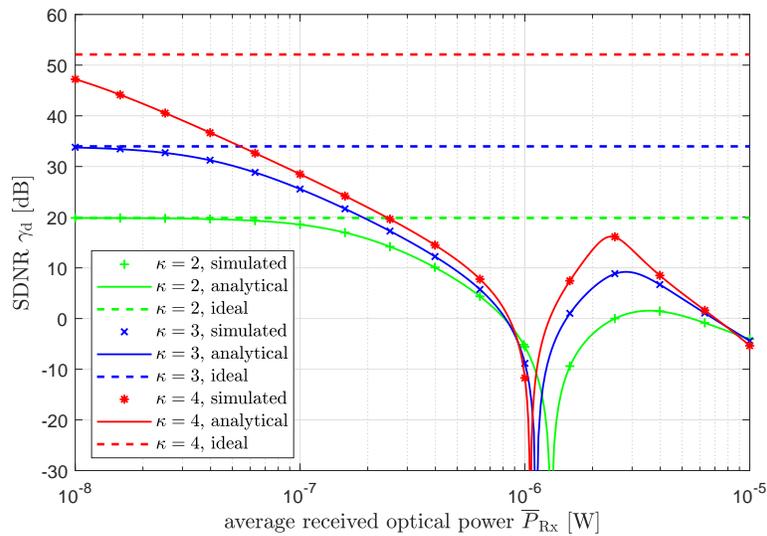}
	\caption{The SDNR versus the received optical power with various clipping level. }\label{SDNR}
\end{figure}
Fig. \ref{SDNR} illustrates the SDNR versus $\overline{P}_{\rm Rx}$ with various $\kappa$. Again the match between analytical and simulation results validates our analytical derivations. When $\overline{P}_{\rm Rx}$ is very small, the nonlinearity of SPAD receiver is negligible, hence its SDNR converges to that of the ideal linear receiver. However, with the increase of $\overline{P}_{\rm Rx}$, the SDNR of SPAD receiver deviates from that of the ideal receiver which remains fixed over $\overline{P}_{\rm Rx}$. This is because for the ideal receiver, the distortion noise is only introduced by the signal clipping at the transmitter and hence the SDNR is irrelevant to the received optical power \cite{Dimitrov12,Chen16}. However, a SPAD receiver suffers from the signal distortion from both clipping and receiver nonlinearity. The change of the $\overline{P}_{\rm Rx}$ causes not only the change of the average electrical signal power, but also the change of the utilized dynamic range of the receiver which introduces a change in distortion-induced noise variance $\sigma_{W_{\rm d}}^2$. Initially, the rise of $\overline{P}_{\rm Rx}$, although it results in higher $\alpha$, it introduces much larger $\sigma_{W_{\rm d}}^2$ and causes the overall reduction of SDNR. When $\overline{P}_{\rm Rx}$ approaches the saturation power, a significant SDNR dip can be observed mainly because of the very small electrical power at this point as mentioned before. The further rise of  $\overline{P}_{\rm Rx}$ in turn increases the SDNR, but eventually the SDNR again diminishes because of the reduction of the electrical signal power in the high power regime. Similar to the ideal receiver, the SPAD receiver also benefits from higher $\kappa$, which can generally provide higher SDNR. The reason is twofold. Firstly, higher $\kappa$ means less distortion noised introduced by the signal clipping at transmitter \cite{Dimitrov12,Chen16}. Secondly, higher  $\kappa$ also indicates smaller scaling factor $\delta$ and hence smaller AC optical signal (see (\ref{xt})). As a result, the received signal occupies smaller dynamic range which leads to less distortion caused by receiver nonlinearity.  

\begin{figure}[!t]
	\centering\includegraphics[width=0.7\textwidth]{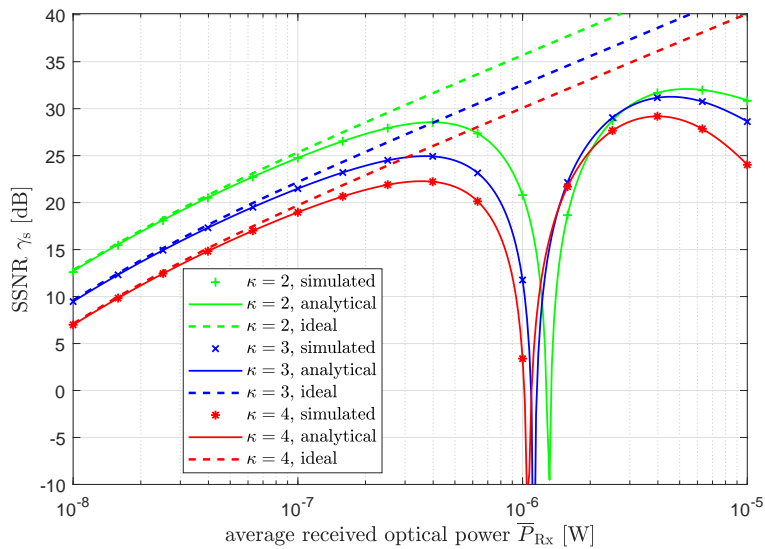}
	\caption{The SSNR versus the received optical power with various clipping level. }\label{SSNR}
\end{figure}
The  SSNR versus $\overline{P}_{\rm Rx}$ with various $\kappa$ is demonstrated in Fig. \ref{SSNR}. Different from the ideal linear receiver whose SSNR monotonically increases with $\overline{P}_{\rm Rx}$, the relationship between SSNR and $\overline{P}_{\rm Rx}$ is more complicated for SPAD receivers. Similar to SDNR, SSNR firstly increases with $\overline{P}_{\rm Rx}$, but a significant SSNR drop can be observed when $\overline{P}_{\rm Rx}$ is close to the saturation power of the SPAD receiver. This is again because of the small electrical signal power at this point. However, different from SDNR, smaller $\kappa$ is generally more preferable which can provide higher SSNR, due to its higher induced electrical signal power as demonstrated in Fig. \ref{alpha_gain}. This implies that there should exist an optimal $\kappa$ that maximizes the total SNR $\gamma$. It is also worth noting that the analytical results of SSNR perfectly match with the simulation results, which validates the Proposition \ref{shotnoise}. 

\begin{figure}[!t]
	\centering\includegraphics[width=0.7\textwidth]{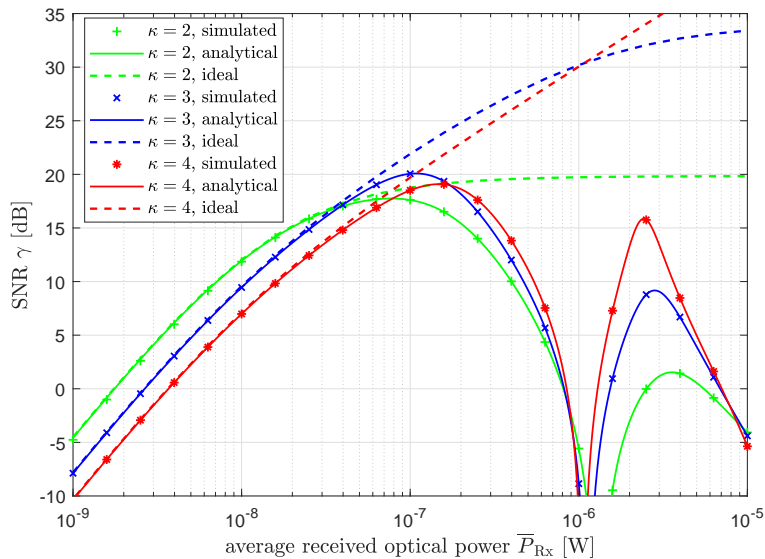}
	\caption{The SNR versus the received optical power with various clipping level. }\label{SNR_figure}
\end{figure}
Fig. \ref{SNR_figure} presents the SNR versus $\overline{P}_{\rm Rx}$ with various $\kappa$. Note that SNR is related to both SDNR and SSNR as given in (\ref{gamma}). It is demonstrated that, different from the ideal receiver whose SNR monotonically increases with $\overline{P}_{\rm Rx}$,  the change of SNR over $\overline{P}_{\rm Rx}$ is similar to SDNR and SSNR discussed above and therefore larger $\overline{P}_{\rm Rx}$ does not always mean higher SNR. It is also illustrated that for different $\overline{P}_{\rm Rx}$, the optimal clipping level is also different. Generally, smaller $\kappa$ provides higher SNR when $\overline{P}_{\rm Rx}$ is low; whereas, larger $\kappa$ in turn gives higher SNR when $\overline{P}_{\rm Rx}$ is high. This attributes to the fact that under different $\overline{P}_{\rm Rx}$, the effect of  either the distortion induced noise or shot noise would be dominant. For $\overline{P}_{\rm Rx}$ where the former is the dominant factor, the SNR is limited by SDNR. In this case, higher $\kappa$ can lead to higher SDNR and hence higher SNR. However, if the shot noise is the dominant factor, SNR is in turn limited by SSNR and less $\kappa$ gives higher SNR. Therefore, for any given $\overline{P}_{\rm Rx}$, the clipping level can be optimized to achieve the optimal performance. In addition, Fig. \ref{SNR_figure} also illustrates that the SNR of SPAD receiver is significantly less than that of the ideal receiver especially when received optical power is high, which shows the significant impact of the SPAD nonlinearity.     

\begin{figure}[!t]
	\centering\includegraphics[width=0.72\textwidth]{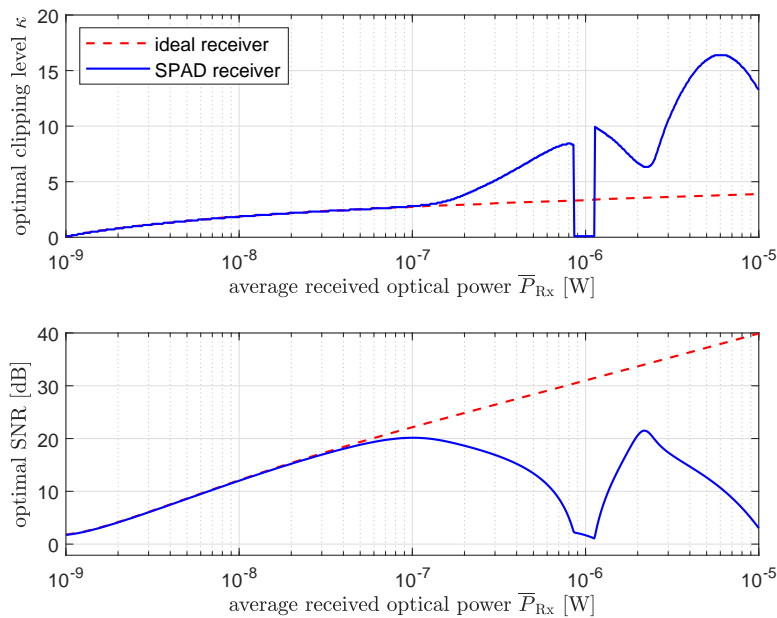}
	\caption{The optimal clipping level and corresponding SNR  versus the received optical power. }\label{optimal_clipping}
\end{figure}
To approach the communication performance limits of the SPAD receiver, in this work, the optimization of signal clipping level $\kappa$ is also investigated through exhaustive search. Fig.  \ref{optimal_clipping} presents the recorded optimal clipping level and the corresponding optimal SNR versus $\overline{P}_{\rm Rx}$. Both SPAD receiver and ideal photon counting receiver are considered. For the ideal receiver, with the increase of $\overline{P}_{\rm Rx}$, the nonlinear distortion induced noise becomes the dominant noise factor. Since the nonlinear distortion is solely introduced by signal clipping which can be sufficiently eliminated when $\kappa$ is above $4$, the optimal $\kappa$ monotonically increases with $\overline{P}_{\rm Rx}$ and finally saturates at around $4$. On the other hand, for the SPAD receiver, the behaviour of optimal $\kappa$ as a function of $\overline{P}_{\rm Rx}$ is more complicated. This is because of the change of the relative amounts of SDNR and SSNR over $\overline{P}_{\rm Rx}$ in the presence of the receiver nonlinearity. 
It is worth noting that when $\overline{P}_{\rm Rx}$ is relatively high, the optimal $\kappa$ of SPAD receiver can be significantly beyond $4$. In fact, when clipping level is above this value, the clipping noise already vanishes and the increase of $\kappa$ is simply for the reduction of the optical signal dynamic range, which can effectively mitigate the nonlinear effect introduced by the receiver. This result indicates that when the receiver nonlinearity is severe, using smaller dynamic range of the transmitter is actually beneficial. 

\begin{figure}[!t]
	\centering\includegraphics[width=0.75\textwidth]{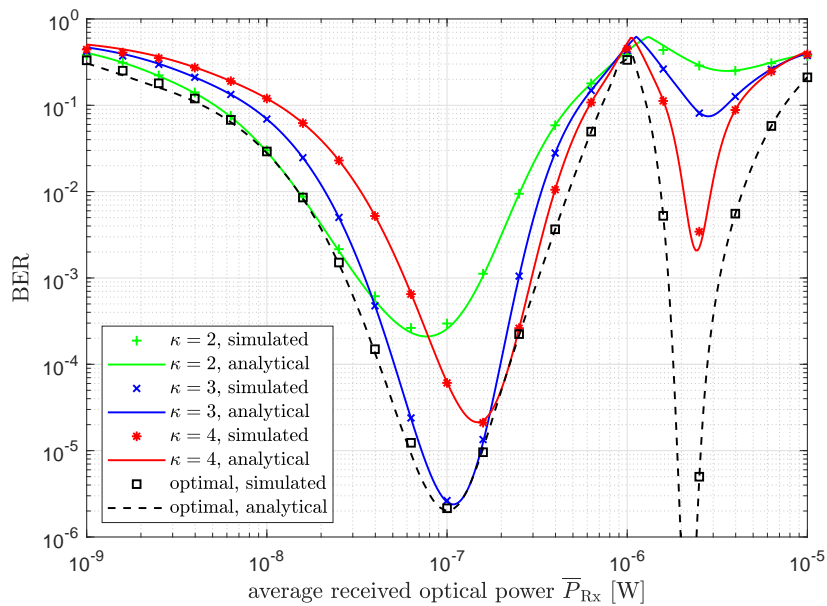}
	\caption{The BER versus the received optical power when $16$-QAM is employed. }\label{BER_figure_M16}
\end{figure}
\begin{figure}[!t]
	\centering\includegraphics[width=0.78\textwidth]{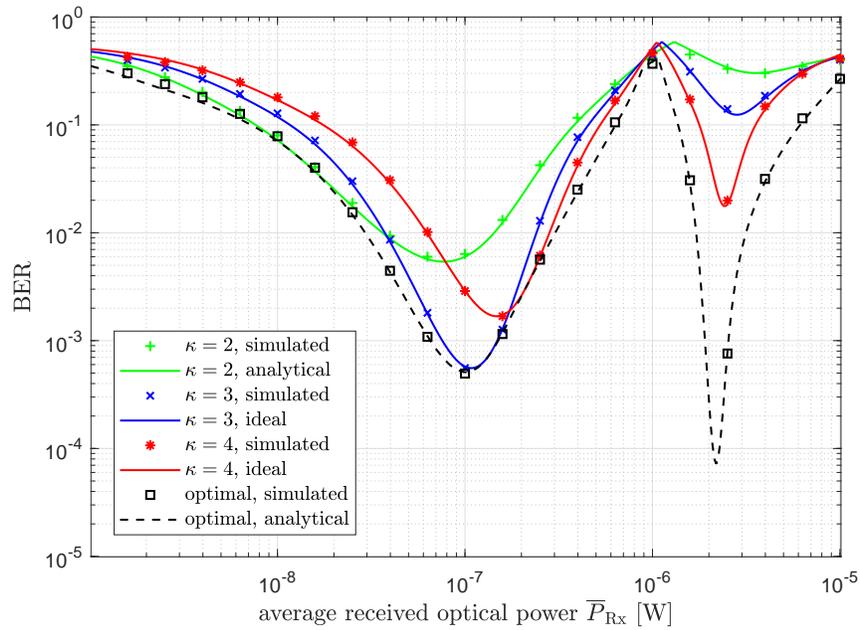}
	\caption{The BER versus the received optical power when $32$-QAM is employed. }\label{BER_figure_M32}
\end{figure}

Let's further consider the BER performance of SPAD-based OWC system with OFDM. The BER versus $\overline{P}_{\rm Rx}$ when $16$-QAM and $32$-QAM are employed are demonstrated in Fig. \ref{BER_figure_M16} and  Fig. \ref{BER_figure_M32}, respectively. Both simulated and analytical results are presented which again match with each other. Due to the two peaks of SNR with the change of $\overline{P}_{\rm Rx}$ shown in Fig. \ref{SNR_figure},  with the increase of $\overline{P}_{\rm Rx}$, two dips of BER can be observed. In addition, different $\kappa$ also results in distinct performances. In these figures, the BER performance when the optimal clipping level is employed is also plotted. It can be observed that by choosing the optimal clipping level, the best BER performance over the whole range of $\overline{P}_{\rm Rx}$ can be achieved. By increasing the modulation order from $16$ to $32$, although the SE increases from $4$ bits/symbol to $5$ bits/symbol,
the BER performance is also degraded. This suggests that in practical implementations, appropriate modulation order should be employed considering the signal SNR and the BER target. 

\begin{figure}[!t]
	\centering\includegraphics[width=0.8\textwidth]{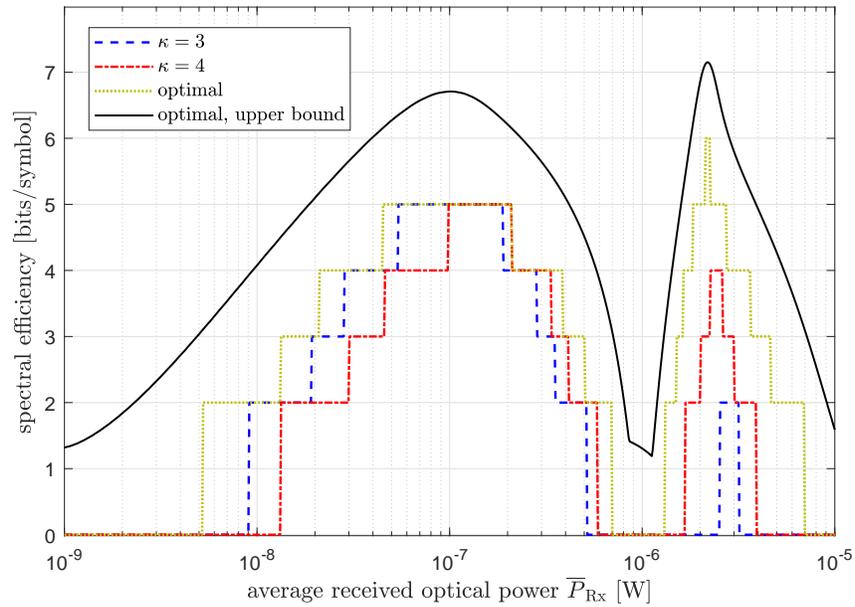}
	\caption{The achievable spectral efficiency versus the received optical power. }\label{data_rate_PR_ana}
\end{figure}
To explore the throughput of the SPAD-based OFDM system, Fig. \ref{data_rate_PR_ana} is plotted which presents the achievable SE versus $\overline{P}_{\rm Rx}$. Note that to achieve this result, a vector of QAM modulation order $[4, 8, 16, 32, 64]$ is considered. For each $\overline{P}_{\rm Rx}$, the BERs when different modulation schemes are employed are calculated according to (\ref{BER}) and the maximal SE with the BER less than a BER target of $3\times10^{-3}$ is recorded. The BER target is selected which is below the $7\%$ forward error correction (FEC) limits \cite{Wang:13}. It is demonstrated in Fig. \ref{data_rate_PR_ana} that the achievable SE varies with $\overline{P}_{\rm Rx}$ and by employing optimal clipping level the highest SE can be achieved over the considered $\overline{P}_{\rm Rx}$, as expected. Therefore, for OWC links with highly dynamic channels, the adaptive signal clipping and modulation scheme can be employed to achieve the best throughput over a wide range of the received power. In addition, the upper bound on achievable SE of the considered system with optimal signal clipping, calculated based on (\ref{SEupper}), is also plotted in Fig. \ref{data_rate_PR_ana}. This upper bound is higher than the SE of the QAM and is the fundamental limit of the investigated system.

\section{Conclusion}\label{con}
SPAD receivers have a great potential for improving the sensitivity of OWC systems and hence combating the outage caused by power fluctuations.  OFDM can be used in SPAD-based OWC systems to improve the SE and boost the data rate. In this work, a theoretical performance analysis of SPAD-based OFDM system is presented. The analytical expressions of both SNR and BER are derived which perfectly match Monte Carlo simulation results verifying the accuracy of the analytical derivations. The presented analytical results provide an effective and accurate way to estimate system performance of practical SPAD-based systems employing OFDM. Through extensive numerical results, the impact of the SPAD nonlinearity on its performance is investigated showing new insights into the best operation regime of PQ SPADs. In addition, the optimization of the signal clipping is further investigated to improve the communication performance. 

\section*{Acknowledgement}
This work is supported by Engineering and Physical Sciences Research Council (EPSRC) under ARROW (EP/R023123/1) and TOWS (EP/S016570/1) grants.

\bibliographystyle{IEEEtran}
\bibliography{IEEEabrv,SPAD_OFDM_Shenjie}
\end{document}